\documentclass{article}[11pt,a4paper,oneside]
\usepackage{epsfig}
\usepackage{amsmath}
\usepackage{amsfonts}
\usepackage{amssymb}
\usepackage{fullpage}
\usepackage{enumerate}
\usepackage{bibspacing}

 \date{}
 \newtheorem{theorem}{Theorem}[section]
\newtheorem{lemma}[theorem]{Lemma}

\newtheorem{corollary}[theorem]{Corollary}

\newenvironment{proof}[1][Proof]{\begin{trivlist}
\item[\hskip \labelsep {\bfseries #1}]}{\end{trivlist}}

 \newtheorem{observation}{Observation}

\title{Characterization of Minimum Cycle Basis in Weighted Partial 2-trees \footnotemark[1]} 
\author{N.S.Narayanaswamy \footnotemark[2]  \and G.Ramakrishna \footnotemark[2]}

\begin{document}
\maketitle 
\footnotetext[1]{A preliminary version of this paper appeared as {\em Characterization of Minimum Cycle Basis in Weighted Partial 2-trees} in the proceedings of CTW 2012}
\footnotetext[2]{Department of Computer Science and Engineering, Indian Institute of Technology Madras, India.  Email: swamy@cse.iitm.ac.in, grama@cse.iitm.ac.in}

\begin{abstract}
\noindent
For a weighted outerplanar graph, the set of lex short cycles is known to be a minimum cycle basis [Inf. Process. Lett. 110 (2010) 970-974 ]. In this work, we show that the set of lex short cycles is a minimum cycle basis in weighted partial 2-trees (graphs of treewidth two) which is a superclass of outerplanar graphs. 
\end{abstract}

\section{Introduction}
\noindent
A cycle basis is a compact description of the set of all cycles of a graph and has various applications including the analysis of electrical networks  \cite{KLM09CycleBasisSurvey}. 
Let $G =(V(G),E(G))$ be an edge weighted graph and let $m=|E(G)|$ and $n = |V(G)|$.
A \emph{cycle} is a connected graph in which the degree of every vertex is two.
An \emph{incidence vector} $x$, indexed by $E(G)$ is associated with every cycle $C$ in $G$, where for every edge $e \in E(G)$, $x_e$ is 1 if $e \in E(C)$ and 0 otherwise.
The \emph{cycle space} of $G$ is the vector space over $\mathbb{F}^{m}_{2}$ spanned by the incidence vectors of cycles in $G$. 
A \emph{cycle basis} of $G$ is a minimum set of cycles whose incidence vectors span the cycle space of $G$.
The weight of a cycle $C$ is the sum of the weights of the edges in $C$.
A cycle basis $B$ of $G$ is a \emph{minimum cycle basis} (MCB) if the sum of the weights of the cycles in $B$ is minimum.  A minimum cycle basis of $G$ is denoted by $MCB(G)$.\\

\noindent
\textbf{Motivation:}
For a weighted graph $G$,  Horton has identified a set $\mathcal{H}$ of $O(mn)$ cycles and has shown that a minimum cycle basis of $G$ is a subset of $\mathcal{H}$ \cite{Horton}. 
Liu and Lu have shown that the set of \emph{lex short cycles} (defined later) is a minimum cycle basis in  weighted outerplanar graphs \cite{Liu_Lu_WeightedOuterPlanar}. We generalize this result for partial 2-trees which is a superclass of outerplanar graphs. \\

\noindent
\textbf{Our contribution:}
The following are the main results in this work.
\begin{theorem}
\label{Theorem_LexCyclesCount}
 Let $G$ be a weighted partial 2-tree on $n$ vertices and $m$ edges. 
 Then the number of lex short cycles in $G$ is $m-n+1$.
\end{theorem}
\begin{theorem}
\label{Theorem_mainResult}
 For a weighted partial 2-tree $G$, the set of lex short cycles is a minimum cycle basis.
\end{theorem}


\noindent
\textbf{Related work:} 
The characterization of graphs using cycle basis was initiated by MacLane \cite{MacLane37}.
In particular, MacLane showed that a graph $G$ is planar if and only if $G$ contains a cycle basis $B$ such that each edge in $G$ appears in at most two cycles of $B$. However, he referred to a cycle basis as a \emph{complete independent set of cycles}.
Formally, the concept of cycle space in graphs was introduced in \cite{Chen71} after four decades.
Later, it was characterized that a planar 3-connected graph $G$ is a Halin graph if and only if $G$ has a planar basis $B$ such that each cycle in $B$ has an external edge \cite{Syslo}.
There after, it was shown that every 2-connected outerplanar graph has a unique MCB \cite{LeydoldOuterPlanarChar}. Subsequently, it was proven that Halin graphs that are not necklaces have a unique MCB \cite{StadlerHalinCB}.

The first polynomial time algorithm for finding an MCB was given by Horton \cite{Horton}.
Since then, many improvements have taken place on algorithms related to minimum cycle basis and its variants.
A detailed survey of various algorithms, characterizations and the complexity status of cycle basis and its variants
was compiled by Kavitha et al.\cite{KLM09CycleBasisSurvey}. The current best algorithm for MCB runs in $O(m^{2}n/\log{n})$ time and is due to Amaldi et al.\cite{Amaldi10MCB}.\\

\noindent
\textbf{Graph preliminaries:}
In this paper, we consider only simple, finite, connected, undirected and weighted graphs.
We refer \cite{dbwestBook} for standard graph theoretic terminologies. 
Let $G$ be an edge weighted graph. 
Let $X \subseteq V(G)$. $G - X$ denotes the graph obtained after deleting the set of vertices in $X$ from $G$. $G[X]$ denotes the subgraph induced by vertices in $X$. $X$ is a \emph{vertex separator} if $G-X$ is disconnected. 
A \emph{component} of $G$ is a maximal connected subgraph.
$K_3$ denotes a cycle on $3$ vertices and $K_2$ denotes an edge.
$K_{2,3}$ is a complete bipartite graph $(V_1,V_2)$ such that $|V_1|=2, |V_2|=3$.
A graph is \emph{planar} if it can be drawn on the plane without any edge crossings.
A planar graph is \emph{outerplanar} if it can be drawn on the plane such that all of its vertices lie on the boundary of its exterior region. 
A 2-tree is defined inductively as follows: $K_3$ is a 2-tree; if $G'$ is a 2-tree and $G = G' \cup \{v\}$ is such that $N_{G}(v)$ forms a $K_2$ in $G$, then $G$ is a 2-tree. A graph is a \emph{partial 2-tree} if it is a subgraph of a 2-tree.
Alternatively, a graph of treewidth (defined in \cite{Robertson_Seymour_1986}) two is a \emph{partial 2-tree}.
An \emph{$H$-subdivision} (or subdivision of $H$) is a graph obtained from a graph $H$ by replacing edges with pairwise internally vertex disjoint paths.

\section{MCB in Weighted Partial 2-trees}
\label{sectionMcb}
For a weighted partial 2-tree $G$ associated with a weight function $w: E(G) \to \mathbb{N}$, we show that the set of lex short cycles (defined below) in $G$ is an $MCB(G)$.
The notion of lex shortest path and lex short cycle is presented from \cite{Hartvigsen1994AMC}.
For a totally ordered set $S$, $\min(S)$ denotes the minimum element in $S$. For a graph $G$, let $V(G)$ be a totally ordered set. A path $P(u,v)$ between two distinct vertices $u$ and $v$ is \emph{lex shortest path} if for all the paths $P'$ between $u$ and $v$ other than $P$, 
 exactly one of the following three conditions hold: 1) $w(P') > w(P)$ 2) $w(P') = w(P)$ and $|E(P')| > |E(P)|$ 3) $w(P') = w(P)$, $|E(P')| = |E(P)|$  and $\min(V(P')\setminus V(P)) > \min(V(P) \setminus V(P'))$, where $w(P) = \Sigma_{e\in E(P)}w(e)$.
  The lex shortest path between any two vertices $u$ and $v$ is unique and is denoted by $lsp(u,v)$. A cycle $C$ is \emph{lex short} if for every two vertices $u$ and $v$ in $C$, $lsp(u,v) \subset C$.   The set of lex short cycles of $G$ is denoted by $LSC(G)$. For a subgraph $G_1$ of $G$, the total order of $V(G_1)$ is the order induced by
the total order of $V(G)$. We use $lsp_{G_1}(x,y)$ to denote the lex shortest path between vertices $x$ and $y$ in $G_1$.
We use the following lemmas from the literature. 

\begin{lemma}[\cite{Hartvigsen1994AMC}]
\label{Lemma_MCBsubSetOfLex}
For a simple weighted graph $G$, $LSC(G)$ contains an $MCB(G)$.
\end{lemma}

\begin{lemma}[\cite{Liu_Lu_WeightedOuterPlanar}]
\label{Lemma_lexCyclesCountOuterPlanar}
 For a simple weighted outerplanar graph $G$, $|LSC(G)| = m-n+1$.
\end{lemma}

\noindent
We present the following lemmas and theorems that are required to prove our main result.

\begin{lemma}
\label{Lemma_existenceOfH}
Let $G$ be a partial 2-tree and $\{u,v\}$ be a vertex separator in $G$. Let $P$ be the lex shortest path between $u$ and $v$.
There exist one component $H$ in $G-\{u,v\}$ such that $V(P) \cap V(H) = \emptyset$ and $E(P) \cap E(H) = \emptyset$.
\end{lemma}
\begin{proof}
If $P=(u,v)$, then none of the components in $G-\{u,v\}$ contain $V(P)$ and $E(P)$. 
If $P=(u,x,v)$, then no component in $G-\{u,v\}$ contain $E(P)$ and exactly one component in  $G-\{u,v\}$ contains $x$.
If $P$ is not captured by these two cases, then $P$ has at least three edges.
If $|E(P)|\geq 3$, then exactly one component in $G-\{u,v\}$ that contains $P-\{u,v\}$.
Since $\{u,v\}$ is a vertex separator in $G$, the number of components in $G-\{u,v\}$ is at least two.  
Therefore, there exist a component $H$ in $G-\{u,v\}$ such that $V(P) \cap V(H) = \emptyset$ and $E(P) \cap E(H) = \emptyset$.
\hfill $\Box$
\end{proof}

\begin{lemma}
\label{LemmaPartial2TreeSep}
Let $G$ be a partial 2-tree that is not outerplanar.  Then there exists a $K_{2,3}(\{u,v\},\{x,y,z\})$-subdivision in $G$ such that $G-\{u,v\}$ contains at least three components.
\end{lemma}
\begin{proof}
A graph  is outerplanar if and only if it contains no subgraph that is a subdivision of $K_4$ or $K_{2,3}$ \cite{GaryBook}.
 Since  a partial 2-tree does not contain a subdivision of $K_4$, a partial 2-tree is outerplanar if and only if it does not contain a subdivision of $K_{2,3}$. Consider a $K_{2,3}(\{u,v\},\{x,y,z\})$-subdivision in $G$. 
Assume that $G-\{u,v\}$ has at most two components.
Then there exist a path in $G-\{u,v\}$ between two vertices in $\{x,y,z\}$ which does not go through the other vertex.
Without loss of generality, we assume that $x$ and $y$ are those two vertices and $z$ is the other vertex. Such a path between $x$ and $y$ is shown as a dotted path in Figure \ref{subdivisionExample}. It follows that there are six internal vertex disjoint paths in $G$, namely $P(x,u),P(x,v),P(y,u),P(y,v),P(x,y)$ and $P(u,v)$ via $z$. 
 Thus, there is a $K_{4}$-subdivision on the vertex set $\{u,v,x,y\}$ in $G$.  This is a contradiction that $G$ is a partial 2-tree. Therefore, $\{u,v\}$ is a vertex separator in $G$ whose removal gives at least three components. $\Box$
\end{proof}

\begin{figure*}[htp!]
\centering
\includegraphics[scale=0.5]{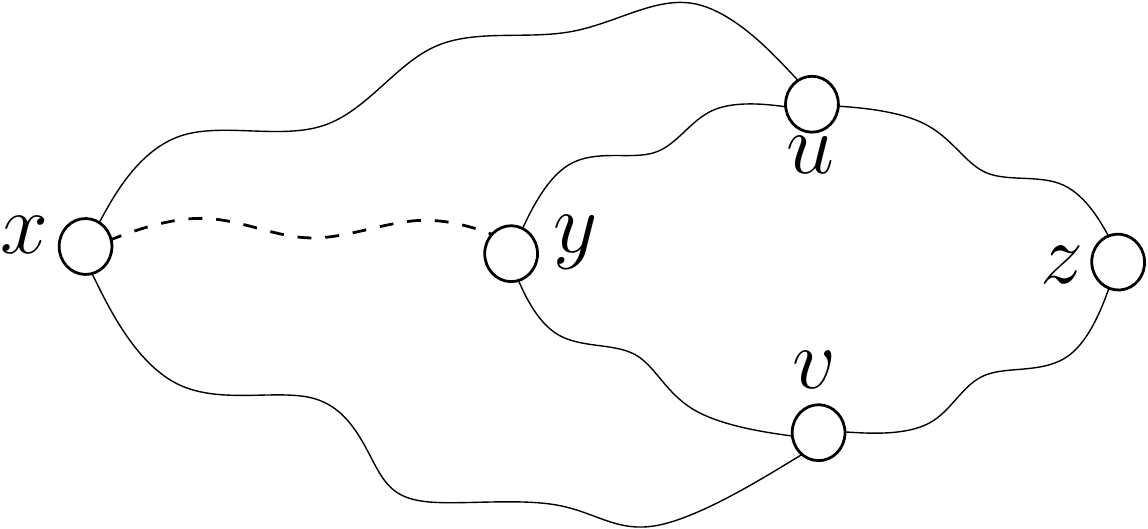}
\caption{A $K_{4}$-subdivision on the vertex set $\{u,v,x,y\}$} 
\label{subdivisionExample}
\end{figure*}

\begin{lemma}
\label{LemmaLexCycle}
Let $G'$ be a weighted subgraph of a weighted graph $G$. Let $P$ be a path and $C$ be a cycle contained both in $G$ and $G'$.\\
 $(a)$. If $P$ in $G$ is lex shortest, then $P$ in $G'$ is  lex shortest.\\
$(b)$. If $C \in LSC(G)$, then $C \in LSC(G')$.
\end{lemma}
\begin{proof}
Suppose if the path $P$ in $G'$ is not lex shortest, then the path $P$ in $G$ would not be lex shortest. Hence, the $P$ in $G'$ is lex shortest.

Recall that $C$ is a lex short cycle if for every $x,y \in V(C)$, $lsp(x,y)$ is contained in $C$. Since $C$ is in $G'$ and for every $x,y \in C$, $lsp(x,y)$ is same as $lsp_{G'}(x,y)$, $C$ is a lex short cycle in $G'$.
$\Box$
\end{proof}

\begin{lemma}
\label{Lemma_IntersectionOfLexShortPaths}
The intersection of two lex shortest paths is either empty or a lex shortest path.
\end{lemma}
\begin{proof}
Consider two lex shortest paths $lsp(x,y)$ and $lsp(u,v)$ in $G$.
Let $G'=(V(G'),E(G'))$, where $V(G') = V(lsp(x,y)) \cap V(lsp(u,v))$, $E(G')=E(lsp(x,y)) \cap E(lsp(u,v))$.
Suppose $V(G') \neq \emptyset$ and $G'$ is not a path,  then we have at least two maximal paths $P(a,b)$ and $P(a',b')$ which are common to both $lsp(x,y)$ and $lsp(u,v)$, where $b \neq a'$. Consequently, the path $P_{1}(b,a')$ contained in $lsp(x,y)$ and the path $P_{2}(b,a')$ contained in $lsp(u,v)$ are different. Since a subpath of a lex shortest path is a lex shortest path, both  $P_1$ and $P_2$ are lex shortest paths between $b$ and $a'$; a contradiction to the fact that for any two vertices in  a graph, there is a unique lex shortest path. 
$\Box$
\end{proof}

\noindent
We decompose a weighted partial 2-tree $G$ that is not outerplanar into two weighted partial 2-trees $G_1$ and $G_2$ in such a way that $LSC(G)$ is equal to the disjoint union of $LSC(G_1)$ and $LSC(G_2)$.
From Lemma \ref{LemmaPartial2TreeSep}, there exist two vertices $u,v\in V(G)$ such that $G - \{u,v \}$ is disconnected and has at least three components. 
Let $P$ be the lex shortest path between $u$ and $v$ in $G$. 
By Lemma \ref{Lemma_existenceOfH}, there exist a component $H$ in $G - \{u,v \}$  such that $V(P) \cap V(H) = E(P)\cap E(H) = \emptyset$. 
The operation $decomp(G,u,v)$ decomposes $G$ into $G_1$ and $G_2$, where 
$V(G_1) = V(H) \cup V(P)$, $E(G_1) = E(H) \cup E(P) \cup \{ (x,y) \mid x \in V(H),y \in \{u,v\} \text{ and } (x,y) \in E(G) \}$,
$G_2 = G[V(G)\setminus V(H)]$. 
An example is shown in Figure \ref{exampleFigure} to illustrate the operation $decomp(G,u,v)$ .\\

\begin{figure*}[htp!]
\centering
\includegraphics[scale=0.5]{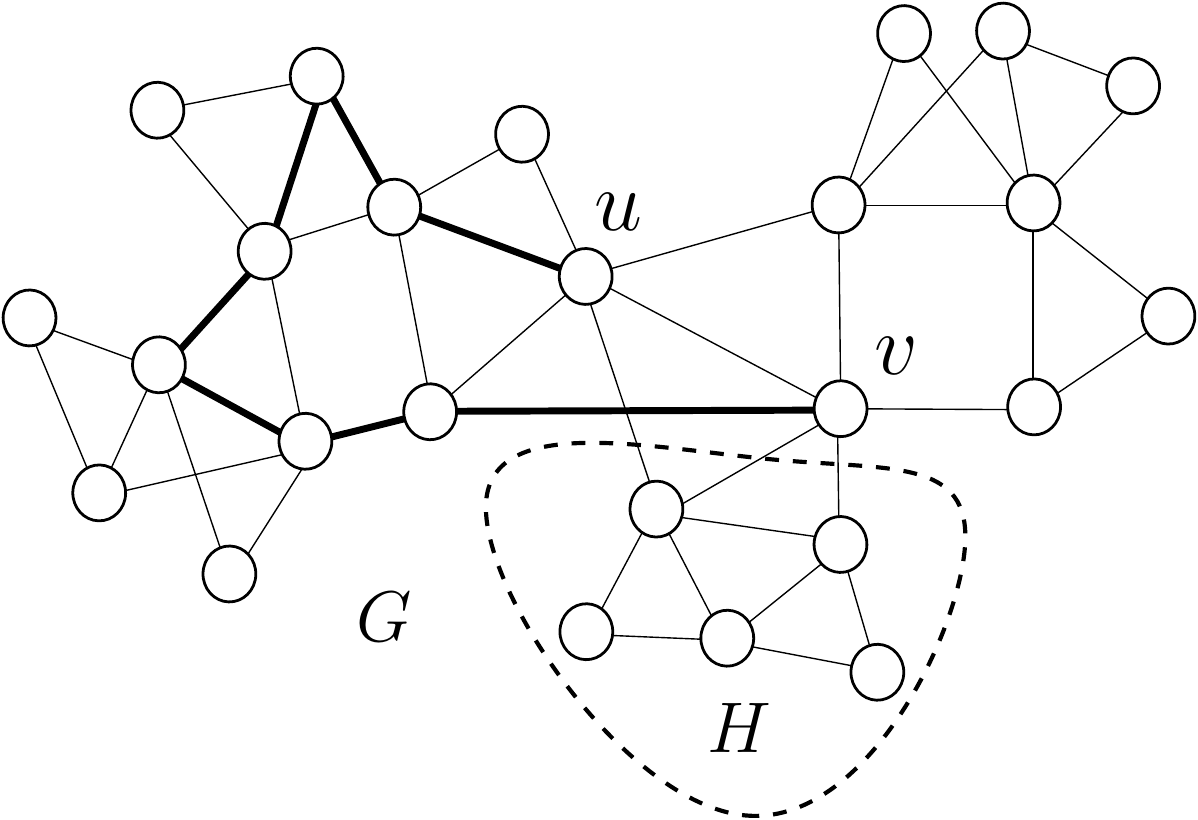}
\includegraphics[scale=0.5]{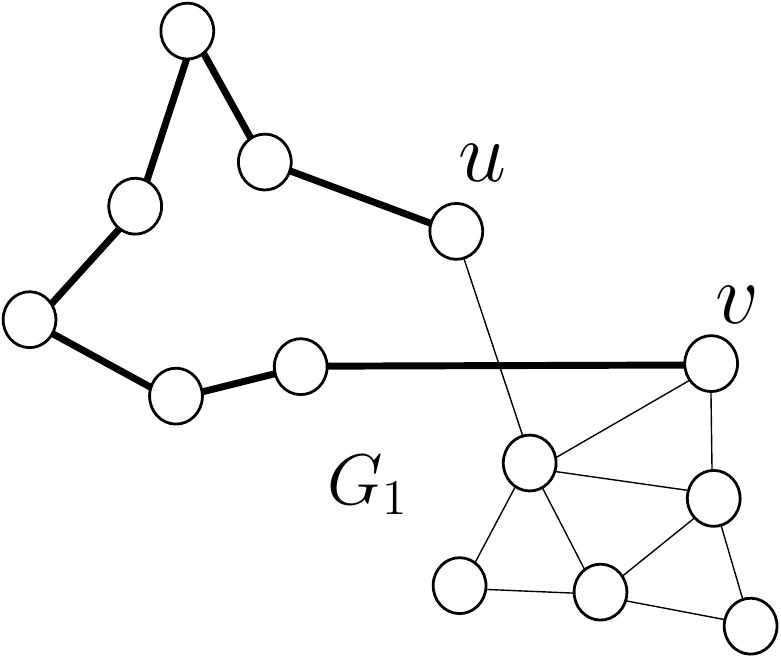}
\includegraphics[scale=0.5]{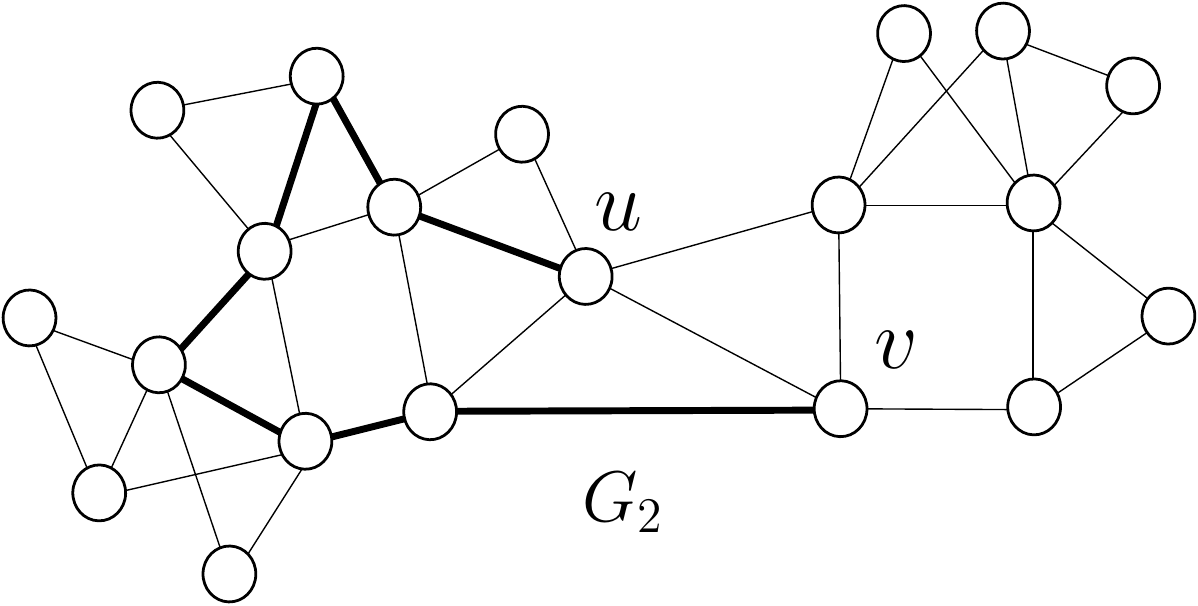}
\caption{For a weighted partial 2-tree $G$, $lsp(u,v)$ is shown in thick edges. Also $H \subset G$ is shown.  $G_1$ and  $G_2$ are the decomposed graphs of $G$.}
\label{exampleFigure}
\end{figure*}

\noindent
We use the following notation for the rest of the paper. $G$ is a weighted partial 2-tree that is not outerplanar.
$\{u,v\}$ is a vertex separator that disconnects $G$ into at least three components.
$H$ is a component in $G-\{u,v\}$ such that $V(lsp(u,v))\cap V(H) = \emptyset$ and $E(lsp(u,v)) \cap E(H) = \emptyset$.
$G_1$ and $G_2$ are the graphs obtained from the operation $decomp(G,u,v)$.\\

\noindent
From the definition of $decomp(G,u,v)$, we have the following two observations. 
\begin{observation}
\label{observation_P1} For $x,y \in V(lsp(u,v))$, $lsp(x,y)$ is in $G_i$. 
\end{observation}
\begin{proof}
 This observation follows, since every subpath of a lex shortest path is a lex shortest path.
 $\Box$
\end{proof}
 
 \begin{observation} 
\label{observation_P2} 
 For $x \in V(G_i)$ and $y \in V(G_i) \setminus V(lsp(u,v))$, every path $P(x,y)$ in $G$ 
 such that the internal vertices of $P$ are in $V(G_i) \setminus \{u,v\}$, is present in $G_i$.
 \end{observation}
 \begin{proof}
Assume that $P(x,y)$ is not in $G_i$. In the path $P$ from $x$ to $y$, let $(a,b)$ be the first edge such that $(a,b) \notin E(P)$. Clearly, $b \notin V(G_i)$. It follows that $a$ is an intermediate vertex in $P$ and $a \in \{u,v\}$; a contradiction to the premise that no intermediate vertex in $P$ belong to $\{u,v\}$. Hence, the observation.
$\Box$
\end{proof}

\begin{lemma}
\label{lemma_everyPairOfVerticesLSP}
For $i \in \{1,2\}$, for every two vertices $x,y \in V(G_i)$, $lsp(x,y)$ is in $G_i$.
\end{lemma}
\begin{proof}
If no vertex is common in $lsp(x,y)$ and $lsp(u,v)$, then from Observation \ref{observation_P2}, $lsp(x,y)$ is in $G_i$.
If at least one vertex is common in $lsp(x,y)$ and $lsp(u,v)$, then due to Lemma \ref{Lemma_IntersectionOfLexShortPaths}, $lsp(x,y) \cap lsp(u,v)$ is $lsp(a,b)$ for some $a,b \in V(lsp(u,v)) \cap V(lsp(x,y))$.
The $lsp(x,y)$ can be viewed as a union of three paths $P(x,a)$, $P(a,b)$ and $P(b,y)$. 
From Observation \ref{observation_P1}, $P(a,b)$ is contained in $G_i$.
If $x = a$, then trivially $P(x,a)$ appears in $G_i$. 
Also, if $y=b$, then clearly $P(b,y)$ appears in $G_i$.
If $x \neq a$, then $x \notin V(lsp(u,v))$.
 Similarly, if $y \neq b$, then $y \notin V(lsp(u,v))$.
From Observation \ref{observation_P2}, it follows that both $P(x,a)$ and $P(b,y)$ appear in $G_i$. 
These observations imply that $lsp(x,y)$ is in $G_i$.
$\Box$
\end{proof}

\begin{corollary}
  \label{Lemma_lexShortCycleReverseDirection}
  For $i \in \{1,2\}$, if there is a cycle $C$ in $LSC(G_i)$, then $C$ is in $LSC(G)$.
\end{corollary}
\begin{proof}
From Lemmas \ref{lemma_everyPairOfVerticesLSP} and \ref{LemmaLexCycle}.$(a)$, for every $x,y \in V(G_i)$, $lsp_{G_i}(x,y)$ and $lsp(x,y)$ are same.
Since $C \in LSC(G_i)$, for every $x,y \in V(C)$, $lsp_{G_i}(x,y)$ is contained in $C$.
Consequently, for every $x,y \in V(C)$, $lsp(x,y)$ is contained in $C$.  Hence $C \in LSC(G)$.
$\Box$
\end{proof}

\begin{theorem}
\label{Theorem_lexCyclesUnion}
 $LSC(G) = LSC(G_1) \uplus LSC(G_2)$.
\end{theorem}
\begin{proof}
 Since $E(G_1) \cap E(G_2)$ is $E(lsp(u,v))$, $LSC(G_1)$ and $LSC(G_2)$ are disjoint.
We now prove that $LSC(G) \subseteq LSC(G_1) \uplus LSC(G_2)$.
Let $C \in LSC(G)$. If $C$ contains at most one vertex from $\{u,v\}$, then $C$ is contained either in $G_1$ or $G_2$, because $\{u,v\}$ is a vertex separator.
Consider the other case when $C$ contains both $u$ and $v$.
Since $C \in LSC(G)$, $C$ contains $lsp(u,v)$. Note that $lsp(u,v)$ is contained both in $G_1$ and $G_2$.
Observe that $E(C) \setminus E(lsp(u,v))$ belongs to $E(G_i)$ for some $i \in \{1,2\}$, because $E(G_1) \cap E(G_2) = E(lsp(u,v))$.
Hence, $C$ is either in $G_1$ or $G_2$.
In both of the cases, by Lemma \ref{LemmaLexCycle}.$(b)$, $C$ is either in $LSC(G_1)$ or $LSC(G_2)$. Therefore, $LSC(G) \subseteq LSC(G_1) \uplus LSC(G_2)$.
From Corollary \ref{Lemma_lexShortCycleReverseDirection}, $LSC(G_1) \uplus LSC(G_2) \subseteq LSC(G)$. Hence, $LSC(G) = LSC(G_1) \uplus LSC(G_2)$.
 $\Box$
\end{proof}

\begin{lemma}
\label{LemmaNumberOfMinors}
The number of $K_{2,3}$-subdivisions in each of $G_1$ and $G_2$ is less than the number of $K_{2,3}$-subdivisions in $G$.
\end{lemma}
\begin{proof}
Recall that there is a $K_{2,3}(\{u,v\},\{x,y,z\})$-subdivision in $G$, and $G_1$ and $G_2$ are obtained from $decomp(G,u,v)$.  Without loss of generality, assume that $x \in V(H)$. Then at most one vertex from $\{y,z\}$ is in $G_1$. Further, observe that $x$ is not in $G_2$. Therefore,  no $K_{2,3}(\{u,v\},\{x,y,z\})$-subdivision exist in $G_1$ and $G_2$. Thus the lemma holds.
$\Box$
\end{proof}

\noindent
\textbf{Proof of Theorem \ref{Theorem_LexCyclesCount}}
\begin{proof}
 We use induction on the number of $K_{2,3}$-subdivisions  in $G$. If the number of $K_{2,3}$-subdivisions in $G$ is zero, then $G$ is outerplanar, since $G$ is a partial 2-tree. From Lemma \ref{Lemma_lexCyclesCountOuterPlanar}, $|LSC(G)| = m-n+1$ when $G$ is outerplanar. Hence base case is true. If $G$ is not an outerplanar graph, then there exists a $K_{2,3}(\{u,v\}, \{x,y,z\})$-subdivision in $G$. 
 From Lemma \ref{LemmaPartial2TreeSep}, $G-\{u,v\}$ is disconnected and contains at least three components. Let $P$ be the $lsp(u,v)$ in $G$ and $k=|V(P)|$. We apply $decomp(G,u,v)$ to obtain $G_1$ and $G_2$ from $G$.
  For $i\in\{1,2\}$, $m_i$ and $n_i$ indicate $|E(G_i)|$ and $|V(G_i)|$, respectively. Now, we can apply induction hypothesis due to Lemma \ref{LemmaNumberOfMinors}.   By induction hypothesis, it follows that $|LSC(G_i)| = m_i - n_i +1$ for $i \in \{i,2\}$. As $P$ is present in $G_1$ and $G_2$, it follows that $n_1 + n_2 = n+k$ and $m_1 + m_2 = m + k -1$. 
 From Theorem \ref{Theorem_lexCyclesUnion}, $LSC(G) = LSC(G_1) \uplus LSC(G_2)$. Hence $|LSC(G)|=|LSC(G_1)|+|LSC(G_2)|$ $=m_1 - n_1 + 1 $ $+ m_2 -n_2 + 1 =m-n+1$. Therefore, $|LSC(G)| = m-n+1$.
$\Box$
\end{proof}

\noindent
\textbf{Proof of Theorem \ref{Theorem_mainResult}}
\begin{proof}
For a simple weighted graph $G$, from Lemma \ref{Lemma_MCBsubSetOfLex}, an $MCB(G) \subseteq LSC(G)$. 
The cardinality of any cycle basis in a graph is known to be $m-n+1$.  For a weighted partial 2-tree $G$, by Theorem \ref{Theorem_LexCyclesCount}, we have $|LSC(G)|=m-n+1$. Therefore, the set of lex short cycles is a minimum cycle basis in weighted partial 2-trees.
$\Box$
\end{proof}

\begin{figure*}[h!]
\centering
\includegraphics[scale=0.3]{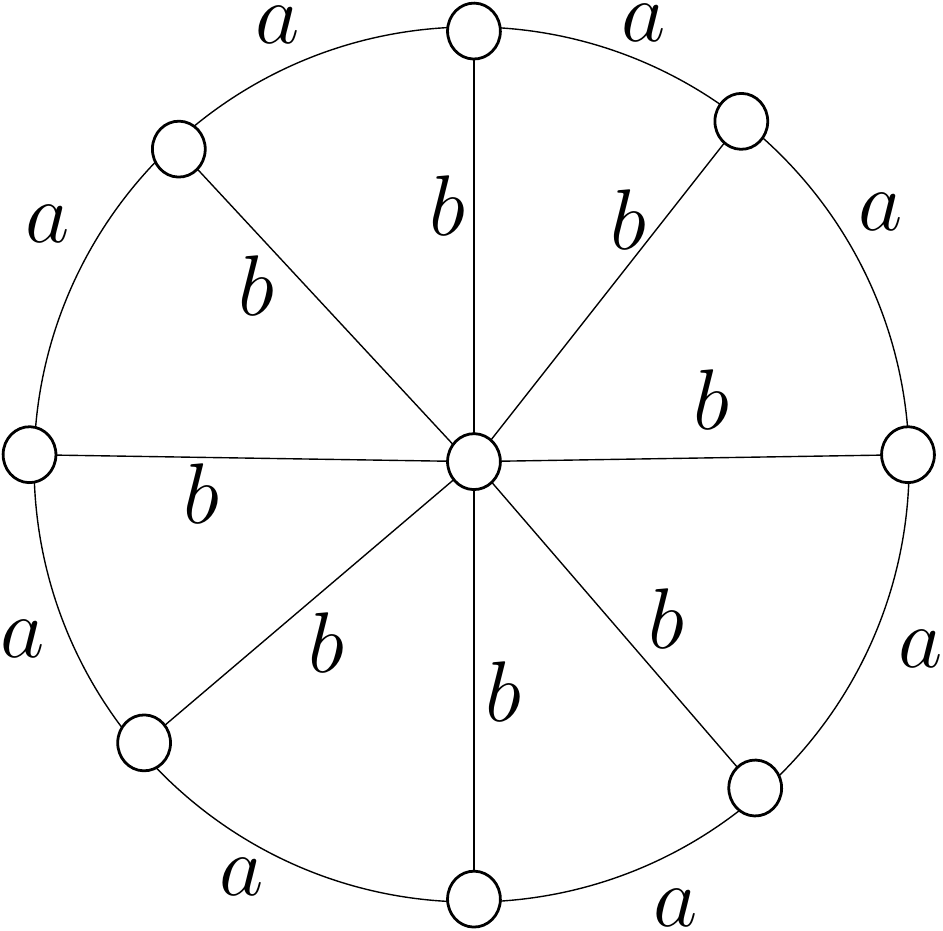}
\caption{For the wheel graph shown, if $b \gg a$, then the set of all triangles and the exterior face are lex short cycles.}
\label{counterExampleFigure}
\end{figure*}

\noindent
We present a family of partial 3-trees for which the set of lex short cycles is not a cycle basis, whose construction is as follows: Let $G_{n}= K_1 + C_{n-1}$ be a wheel graph on $n$ vertices, where  $n \geq 4$.
 A wheel graph on 9 vertices is depicted in Figure \ref{counterExampleFigure}.
Note that $G_n$ is planar.  For every edge $e\in E(G_n)$, assign $w(e)=a$ if $e$ is in external face, otherwise $w(e) = b$, where $a,b \in \mathbb{N}$ and $b \gg a$. Since every face in $G_n$ is a lex short cycle, $|LSC(G_n)|=m-n+2$ by Euler's formula. Therefore, $LSC(G_n)$ can not be a cycle basis.
\bibliographystyle{plain}

\bibliography{sp}

\end{document}